\documentclass[lettersize,journal]{IEEEtran}

\usepackage{soul}
\usepackage{amsmath,amsfonts}
\usepackage{algorithmic}
\usepackage{array}
\usepackage{textcomp}
\usepackage{stfloats}
\usepackage{url}
\usepackage{verbatim}
\usepackage{graphicx}
\def\BibTeX{{\rm B\kern-.05em{\sc i\kern-.025em b}\kern-.08em
		T\kern-.1667em\lower.7ex\hbox{E}\kern-.125emX}} 
\usepackage{balance} 

\usepackage{xargs}                  
 \usepackage{latexsym}
\usepackage[table,xcdraw]{xcolor}
\usepackage{tabularx}
\usepackage{amssymb, amsthm}
\usepackage{placeins}
\usepackage[font=footnotesize,labelfont=bf]{caption}
\usepackage{subcaption}
\usepackage{mathtools}
\usepackage{multirow}
\usepackage{cite, enumerate}
\usepackage[inline,shortlabels]{enumitem}
\usepackage{dsfont}
\usepackage{siunitx} 
\usepackage{stackengine,scalerel}
\usepackage{mathrsfs, eucal}
\usepackage{bm}
\usepackage{accents}
\usepackage[normalem]{ulem}
\usepackage{tikz}
\usepackage{circuitikz}
\usepackage{tikz-3dplot}
\usetikzlibrary{shapes, shapes.geometric, arrows, positioning}
\usepackage{tkz-euclide}
\tikzset{
    bs/.pic = {                                      
        \draw[line width = 1pt,-round cap] (0,0.4\R) -- (-0.2\R,-0.2\R);
        \draw[line width = 1pt,-round cap] (0,0.4\R) -- (0.2\R,-0.2\R);
        \draw[line width = 1pt] (-0.2\R,-0.2\R) -- (0.133\R,0);
        \draw[line width = 1pt] (0.2\R,-0.2\R) -- (-0.133\R,0);
        \draw[line width = 1pt,-round cap] (-0.133\R,0) -- (0.067\R,0.2\R);
        \draw[line width = 1pt,-round cap] (0.133\R,0) -- (-0.067\R,0.2\R);
        \draw[line width = 1pt,-round cap] (-0.213\R,0.4\R) -- (0.213\R,0.4\R);
        \draw[line width = 1pt,-round cap] \foreach \x in {-0.21, -0.14,...,0.22} {(\x\R,0.4\R) -- (\x\R,0.47\R)};
        \node (dim) at (0,0)  [align=center,minimum width=0.5\R,minimum height=1\R] {};
    },
}
\tikzset{
    user/.pic = {     
        \draw[rounded corners=0.02\R] (-0.09\R,-0.17\R) rectangle (0.09\R,0.17\R) ;                     
        \draw[fill=gray] (-0.08\R,-0.12\R) rectangle (0.08\R,0.12\R);                          
        \draw[rounded corners=0.005\R] (-0.04\R,0.14\R) rectangle (0.04\R,0.15\R);                     
        \draw (0,-0.14\R) circle (0.012\R) ; 
        \node (dim) at (0,0)  [align=center,minimum width=0.2\R,minimum height=0.3\R] {};
    }
}

\tikzset{
  drone/.pic = {
    \begin{scope}[scale=0.3]
      \fill[black] (-1.2,0) to[out=5,in=175] (1.2,0)
                     -- (0.7,-0.3) -- (-0.7,-0.3) -- cycle;
      \foreach \x in {-1.2,1.2} {
        \draw[very thick, black, line cap=round] (\x,0.1) -- (\x,0.4); 
        \draw[very thick, black, line cap=round] (\x-0.5,0.4) -- (\x+0.5,0.4); 
        \fill[black] (\x,0.4) circle (0.06); 
      }
      \fill[black] (-0.5,-0.3) -- (-0.9,-1) -- (-0.7,-1) -- (-0.3,-0.3) -- cycle;
      \fill[black] (0.5,-0.3) -- (0.9,-1) -- (0.7,-1) -- (0.3,-0.3) -- cycle;
      \fill[black] (-0.25,-0.75) rectangle (0.25,-1.25);
      \fill[red] (0,-1.0) circle (0.08);
    \end{scope}
  }
}

\tikzset{
    target/.pic = {
    \begin{scope}[scale=0.6]
        \tikzset{every path/.style={line width=1.3pt}}
        \draw
            (-1,0) -- (-0.75,0)
            .. controls (-0.65,0) and (-0.65,-0.1) .. (-0.6,-0.1)
            arc[start angle=270, end angle=450, radius=0.2cm]
            .. controls (-0.55,-0.1) and (-0.55,0) .. (-0.45,0)
            -- (0.45,0)
            .. controls (0.55,0) and (0.55,-0.1) .. (0.6,-0.1)
            arc[start angle=270, end angle=450, radius=0.2cm]
            .. controls (0.65,-0.1) and (0.65,0) .. (0.75,0)
            -- (1,0) -- (1,0.5) -- (-1,0.5) -- cycle;
        \draw[line width=1pt] (-0.6,0.5) -- (-0.4,0.9) -- (0.4,0.9) -- (0.6,0.5);
        \draw[line width=1pt] (-0.3,0.55) rectangle (0.0,0.85);
        \draw[line width=1pt] (0.1,0.55) rectangle (0.4,0.85);
        \fill (-0.6,-0.1) circle (0.02);
        \fill (0.6,-0.1) circle (0.02);
        \end{scope}
    }
}

\tikzset{
    repeater/.pic={
        \begin{scope}[scale=.8]
            \tikzset{every path/.style={line width=1.3pt}}
            \draw (-.5,0) rectangle (.5,.5);
            \draw (-.5,0) -- (-.5,.8);   
            \draw (.5,.5) -- (.5,.8);     

            \draw (-.7,1) -- (-.5,.8) -- (-0.3,1) -- cycle;
	     \draw (.3,1) -- (.5,.8) -- (0.7,1) -- cycle;
  	     \draw[->,thick, blue] (-.5,1) .. controls (-.3,2) and (.3,2) .. (.5,1) node [xshift= -4mm, yshift=3mm]{$\nu$};
        \end{scope}
    }
}

\DeclareMathAlphabet\mathbfcal{OMS}{cmsy}{b}{n}

\usepackage{nicematrix}
\usepackage{arydshln}

\usepackage{acronym}
\usepackage[shortcuts]{glossaries}
\newacronym{dl}{DL}{downlink}
\newacronym{mimo}{MIMO}{multiple-input multiple-output}
\newacronym{csi}{CSI}{channel state information}
\newacronym{bs}{BS}{base station}
\newacronym{ap}{AP}{access-point}
\newacronym{tdd}{TDD}{time-division duplex}
\newacronym{awgn}{AWGN}{additive white Gaussian noise}
\newacronym{ml}{ML}{maximum-likelihood}
\newacronym{psd}{PSD}{positive semi-definite}
\newacronym{iid}{i.i.d.}{independent and identically distributed}
\newacronym{rcs}{RCS}{radar cross-section}
\newacronym{evd}{EVD}{Eigen value decomposition}
\newacronym{svd}{SVD}{singular value decomposition}
\newacronym{ls}{LS}{least-square}
\newacronym{ils}{ILS}{iterative least-square}
\newacronym{cdf}{CDF}{cumulative distribution function}
\newacronym{uav}{UAV}{unmanned aerial vehicle}
\newacronym{wrt}{w.r.t}{with respect to}
\newacronym{aod}{AoD}{angle of departure}
\newacronym{aoa}{AoA}{angle of arrival}
\newacronym{music}{MUSIC}{multiple signal classification}
\newacronym{los}{LoS}{line-of-sight}
\newacronym{mle}{MLE}{maximum-likelihood estimator}
\newacronym{pdf}{pdf}{probability density function}
\newacronym{isac}{ISAC}{integrated sensing and communication}
\newacronym{sinr}{SINR}{signal-to-interference-plus-noise ratio}
\newacronym{glrt}{GLRT}{generalized likelihood ratio test}
\newacronym{zf}{ZF}{zero-forcing}
\newacronym{ris}{RIS}{reconfigurable intelligent surface}
\newacronym{ncr}{NCR}{network-controlled repeater}
\newacronym{se}{SE}{spectral-efficiency}
\newacronym{gpp}{3GPP}{3rd Generation Partnership Project}
\newacronym{pod}{PoD}{probability of detection}
\newacronym{pfa}{PoFA}{probability of false alarm}
\newacronym{pomd}{PoMD}{probability of missed detection}
\newacronym{nlos}{NLoS}{non line-of-sight}
\newacronym{ofdm}{OFDM}{orthogonal frequency division multiplexing}
\usepackage[hidelinks, breaklinks, allcolors=blue]{hyperref}
\hypersetup{
	colorlinks   = false, 
	urlcolor     = cyan, 
	linkcolor    = magenta, 
	citecolor   = magenta, 
}
\usepackage[capitalise,nameinlink]{cleveref}  
\usepackage{prettyref}
\crefformat{equation}{(#2#1#3)} 
\crefname{figure}{Fig.}{Figs.}
\crefname{section}{Sec.}{Secs.}
\crefname{algocf}{Algo.}{Algos.}
\Crefname{algocf}{Algorithm}{Algorithms}
\usepackage{float}
\usepackage[linesnumbered, ruled, vlined]{algorithm2e}
\SetCommentSty{mycommfont}

\SetCommentSty{xCommentSty}
\LinesNumbered
\SetSideCommentRight
\DontPrintSemicolon
\RestyleAlgo{algoruled}

\SetCommentSty{mycommfont}
\SetKwInput{KwInput}{Input}                
\SetKwInput{KwOutput}{Output}            
\SetKwInput{KwInt}{Initialization}            

\newenvironment{myfigure*}
{\begin{figure*}[!t]\centering}
	{\hrule\end{figure*}}

\newcommand{\mt}[1]{\mathsf{#1}}
\newcommand{\mtt}[1]{{\mathtt{ #1}}}
\newcommand{\mr}[1]{{\mathrm{ #1}}}

\newcommand{\Elr}[1]{{ \mathbb{E} \left[#1  \right] }}

\newcommand{\bpr}[1]{{\left( #1 \right)}}

\newcommand{\bsr}[1]{{\left[ #1 \right]}}
\newcommand{\bbr}[1]{{\left\{ #1 \right\}}}

\newcommand{\vt}[1]{{\mtt{Vec} \left( #1 \right)}}

\newcommand{\snorm}[1]{{ \left\Vert #1 \right\Vert_{2}^2 }}

\newcommand{\vect}[1]{{\mathbf{ #1}}}

\newcommand{\ocirc}[1]{\ThisStyle{\ensurestackMath{%
  \stackon[.3pt]{\SavedStyle#1}{\SavedStyle\kern.05\LMpt\circ}}}}

\newcommand{\abs}[1]{{\big\lvert {#1}\big\rvert}}

\newcommand{\cn}[1]{{\sim\mathcal{CN}\left({#1}\right)}}
\newcommand{\cno}[1]{{\mathcal{CN}\left({#1}\right)}}
\newcommand{\mbbC}[1]{{\in\mathbb{C}^{{#1}}}}
\newcommand{\argmax}[1]{{\underset{{#1}}{\mr{arg\,max}}}}

\newcommand{\maximize}[1]{{\underset{\bbr{#1}}{\mr{max}}}}
\newcommand{\minimize}[1]{{\underset{\bbr{#1}}{\mr{min}}}}
\newcommand{\p}[1]{p\left({#1}\right)}

\newcommand\gldec[2]{
	\underset{#1}{\overset{#2}{\gtrless}}
}

\def\CalH{\mathcal{H}}

\def\CalL{\mathcal{L}}

\def\CalT{\mathcal{T}}

\def\vB{\vect{B}}
\def\vc{\vect{c}}
\def\vC{\vect{C}}

\def\vf{\vect{f}}

\def\vG{\vect{G}}

\def\vI{\vect{I}}

\def\vp{\vect{p}}

\def\vQ{\vect{Q}}
\def\vr{\vect{r}}
\def\vR{\vect{R}}

\def\vU{\vect{U}}

\def\vw{\vect{w}}

\def\vx{\vect{x}}

\def\vy{\vect{y}}

\def\vz{\vect{z}}
\def\bba{\pmb{a}}

\def\bbb{\pmb{b}}

\def\bbSigma{\pmb{\Sigma}}

\def\sfc{{\mt{c}}}

\def\sfg{{\mt{g}}}

\def\sfr{{\mt{r}}}
\def\sfs{{\mt{s}}}
\def\sft{{\mt{t}}}

\def\mttB{\mtt{B}}
\def\mttC{\mtt{C}}

\def\mttI{\mtt{I}}

\def\mttL{\mtt{L}}

\def\mttO{\mtt{O}}

\def\mttR{\mtt{R}}

\def\mttT{\mtt{T}}

\def\Nr{N_{\mr{r}}}
\def\Nt{N_{\mr{t}}}

\def\NvarB{\sigma^2_{\mathtt{BS}}}
\def\NvarR{\sigma^2_{\mathtt{R}}}
\def\NvarUE{\sigma^2_{\mathtt{UE}}}
\def\vZ{\mathbf{0}}

\def\ie{\text{i.e.}~}

\def\Nt{N_{\sft}}
\def\Nr{N_{\sfr}}

\def\att{\bba_{\sft}}
\def\atr{\bba_{\sfr}}
\def\btt{\bbb_{\sft}}
\def\btr{\bbb_{\sfr}}
\def\sinr{\mtt{SINR}}

\newtheorem{thm}{Theorem}
\newtheorem{lem}[thm]{Lemma}

\newtheorem{rem}{Remark}

\newcounter{appctr}
\renewcommand{\theappctr}{\Alph{appctr}}
\newcommand{\appsection}[1]{%
  \refstepcounter{appctr}%
  \section*{Appendix \theappctr: #1}%
  \addcontentsline{toc}{section}{Appendix \theappctr: #1}%
}

\newcommand{\edit}[1]{{\textcolor{black}{#1}}}
\graphicspath{{./Figures/}} 
\allowdisplaybreaks
\def\BibTeX{{\rm B\kern-.05em{\sc i\kern-.025em b}\kern-.08em
    T\kern-.1667em\lower.7ex\hbox{E}\kern-.125emX}}
\usepackage{balance}

\let\mybibitem\bibitem
\renewcommand{\bibitem}[1]
{\ifstrequal{#1}{Wang_TWC}{\color{black}\mybibitem{#1}}
{\ifstrequal{#1}{Xiao_TIT}{\color{black}\mybibitem{#1}}
{\ifstrequal{#1}{Bazzi_Chafii}{\color{black}\mybibitem{#1}}
{\color{black}\mybibitem{#1}}
}}}

\begin{document}
\title{On the Performance of Dual-Antenna Repeater Assisted Bi-Static MIMO ISAC}
\author{Anubhab Chowdhury and Erik G. Larsson,~\textit{Fellow,~IEEE}
\thanks{The authors are with the Department of Electrical Engineering (ISY), Linköping University, 58183 Linköping, Sweden. Emails: {\{anuch87, erik.g.larsson\}{@liu.se}}}
\thanks{This work was supported in part by the Knut and Alice Wallenberg (KAW) foundation, Excellence
Center at Linköping-Lund in Information Technology (ELLIIT), and the Swedish Research Council (VR).}
}


\maketitle

\begin{abstract}
This paper presents a framework for target detection and \gls{dl} data transmission in a repeater-assisted bi-static integrated sensing and communication system. A repeater is an active scatterer that retransmits incoming signals with a complex gain almost instantaneously, thereby enhancing sensing performance by amplifying the echoes reflected by the targets. The same mechanism can also improve \gls{dl} communication by mitigating coverage holes. However, the repeater introduces noise and increases interference at the sensing receiver, while also amplifying the interference from target detection signals at the \gls{dl} users. The proposed framework accounts for these sensing-communication trade-offs and demonstrates the potential benefits achievable through a carefully designed precoder at the transmitting base station. In particular, our finding is that a higher value of probability of detection can be attained with considerably lower target radar-cross-section variance by deploying repeaters in the target hot-spot areas.
\end{abstract}

\begin{IEEEkeywords}
Repeater, Integrated sensing and communication, Massive MIMO, GLRT
\end{IEEEkeywords}

\section{Introduction}
\IEEEPARstart{R}{epeaters}~{(also known as full-duplex relays) are active scatterers,  capable of receiving and retransmitting signals instantly with amplification \cite{Jianan_TWC}.}
They can be deployed on a large scale to enhance coverage for \gls{mimo} wireless systems~\cite{Erik_WCL, Sara_Vieira_Erik_Mag, Jianan_TWC, Byoung_Magazine, Sergiy_WCL, Emil_repeater}.
Reference~\cite{Sara_Vieira_Erik_Mag} demonstrated that repeaters can procure a \gls{se}  similar to that of distributed \gls{mimo}, while obviating the front-haul signalling overhead. Furthermore,~\cite{RIS_Makki_IEEE_Acess} indicated that repeaters can outperform \gls{ris}-assisted systems, with the additional advantage of repeaters being band/frequency selective. The attractive implementation aspects of repeaters have led to their inclusion in the 5G NR standards since \gls{gpp} Release $18$~\cite{3gpp}.

On parallel avenue, \gls{isac} offers efficient spectral and hardware utilization, and is actively being pushed forward as part of \gls{gpp} Release $19$, a 5G-advanced/pre-6G initiative~\cite{ISAC_6G}. A detailed survey of \gls{isac} is beyond the scope of this paper; readers can refer to~\cite{ISAC_survey, Zhang_Survey_2022, Michail_TWC} and references therein, and specifically for bi-static \gls{isac},~\cite{Wang_TWC, Xiao_TIT, Bazzi_Chafii}. \edit{Specifically,~\cite{Bazzi_Chafii} optimized \gls{ris}-assisted bistatic sensing.}
The authors in~\cite{Erik_EUSIPCO} considered repeater-assisted \emph{mono-static} \gls{isac},\footnote{In a mono-static setting, the \gls{bs} is required to be full-duplex, which in turn incurs additional signal processing and hardware overhead for self-interference cancellation.} where the repeater was used to improve the coverage for \gls{dl} communication users.
{In contrast to~\cite{Erik_EUSIPCO}, we investigate the feasibility of \emph{bi-static}~\gls{isac} in the presence of a repeater (specifically,  a dual-antenna repeater~\cite[Fig. $2$]{Sara_Vieira_Erik_Mag}),  focusing on the effects of the repeater on the performance of both target detection and \gls{dl} communication.}

In bi-static \gls{isac},  one of the \glspl{bs} transmits precoded \gls{dl} data and a dedicated signal for target detection. The other \gls{bs} performs joint clutter-channel and \gls{rcs} estimation, and target detection, using the reflected echoes. 
In this setting, \emph{ideally}, a repeater can enhance \gls{isac} performance in two ways: (a) by amplifying the echo reflected from the target, which in turn can improve the detection performance; (b) by removing coverage-holes for the \gls{dl} users. 

However, the repeater also amplifies communication-to-sensing and sensing-to-communication interference. Specifically, the amplified sensing signal can incur increased interference for the communication users. Contrarily, for sensing, should the amplified \gls{dl} signals~(directly via the repeater) overwhelm the reflected amplified echo of the target, detection performance at the sensing receiver can severely degrade. {To this end, this paper investigates these tradeoffs and discusses the \emph{choices of precoders} at the transmit \gls{bs} that can mitigate this interference while retaining the \gls{isac} performance improvement due to the repeater. Moreover, the \emph{developed \gls{glrt} framework} is agnostic of the channel models and applicable to diverse propagation conditions among the repeater, BS, and users. }

\begin{figure}
	\centering
	\begin{subfigure}{\columnwidth}
		\centering		
		\begin{tikzpicture}[scale=0.75, transform shape]
		\newdimen\R
		\R=2cm
		\draw (-4,3) pic (BSt) {bs} node [xshift = 0mm, yshift=12mm] {$\vx\mbbC{\Nt}$};
		\draw (1.7,1) pic (UE) {user};
		\draw (2.4,2) pic (UE2) {user};		
		\draw (2,3.1) pic (UE3) {user} node [yshift=-5mm]{$n$th User};
		\draw (-2,1) pic (R) {repeater} node [yshift=-1.9mm]{Repeater};
		\draw (-4,1) pic (Trg) {target} node [yshift=-2.7mm]{Target};
		\node  [align=center,below=-4mm of BStdim.south] {Tx. \gls{bs}};
		
		\draw[-{Stealth[length=3.5mm,  open, round]},black, dashed, thick,shorten >= 0.3cm] (-4,3.8) to (-2.2,1.7) node [xshift=-5mm, yshift=13mm]{$\btt^{T}$};
		
		\draw[-{Stealth[length=3.5mm,  open, round]},red, thick,shorten >= 0.3cm] (-4,1.6) to (-2.2,1.7);
		
		\draw[-{Stealth[length=3.5mm,  open, round]},black, dashed, thick,shorten >= 0.3cm] (-4,3.8) to (2,3.1) node [xshift=-15mm, yshift=0mm]{$\vf_{n}$};
		
		\draw[-{Stealth[length=3.5mm,  open, round]},black, dashed, thick,shorten >= 0.3cm] (-1.4,1.9) to (2,3.1) node [xshift=-20mm, yshift=-4mm]{$h_{n}$};

\draw[-{Stealth[length=3.5mm,  open, round]},black, dashed, thick,shorten >= 0.3cm] (-1.4,1.9) to (2.4,2);

\draw[-{Stealth[length=3.5mm,  open, round]},black, dashed, thick,shorten >= 0.3cm] (-1.4,1.9) to (1.8,1.4);
		\end{tikzpicture}
	\caption{Bi-static massive \gls{mimo} \gls{isac}: \emph{communication view-point}. Multi-path reflections from the target~(not via the repeater) can be incorporated in $\vf_n$, which can constitute both \gls{los} and \gls{nlos} components.}\label{fig: system_model_comm}
	\end{subfigure}
		\begin{subfigure}{\columnwidth}
		\centering
		\begin{tikzpicture}[scale=0.75, transform shape]
		\newdimen\R
		\R=2cm
		\draw (-4,3) pic (BSt) {bs} node [xshift = 0mm, yshift=12mm] {$\vx\mbbC{\Nt}$};
		\draw (2,3) pic (BSr) {bs};
		\draw (-4.7,1) pic (UE) {user};
		\draw (-5.4,2) pic (UE2) {user};		
		\draw (-5,3.1) pic (UE3) {user} node [yshift=-5mm]{User};
		\draw (-1.9,1.8) pic (Trg) {target} node [yshift=-2.7mm]{Target};
		\draw (.5,1) pic (R) {repeater} node [yshift=-1.9mm]{Repeater};
		\node  [align=center,below=-4mm of BStdim.south] {Tx. \gls{bs}};
		\node  [align=center,below=-4mm of BSrdim.south] {Rx.  \gls{bs}};
		\draw[-{Stealth[length=3.5mm,  open, round]},black, dashed, thick,shorten >= 0.3cm] (-4,3.8) to (-1.5,2.2) node [xshift=-13mm, yshift=3mm]{$\att^{T}$};
		\draw[-{Stealth[length=3.5mm,  open, round]},black, dashed, thick,shorten >= 0.3cm] (-1.9,2.3) to (1.7,4) node [xshift=-17mm, yshift=-5mm]{$\atr\mbbC{\Nr}$};
		\draw[-{Stealth[length=3.5mm, green, open, round, fill=white]},green, thick,shorten >= 0.3cm] (-1.9,2.4) to (.4,1.9) node [xshift=-10mm, yshift=-1mm]{${\color{black}\sfg\in\mathbb{C}}$};
		\draw[-{Stealth[length=3.5mm,  open, round]},red, dashed, thick,shorten >= 0.3cm] (-4,3.8) to (.4,1.9) node [xshift=-24mm, yshift=13mm]{$\btt^{T}$};
		\draw[-{Stealth[length=3.5mm,  open, round]},blue, dashed, thick,shorten >= 0.3cm] (.9,1.9) to (1.7,4) node [xshift=-5mm, yshift=-8mm]{$\btr$};
		\draw[-{Stealth[length=3.5mm,  open, round]},red, dashed, thick,shorten >= 0.3cm] (-4,3.8) to (1.7,3.9) node [xshift=-27mm, yshift=1.4mm]{$\vG_{\mttB}$};
		\end{tikzpicture}
	\caption{Bi-static massive \gls{mimo} \gls{isac}: \emph{sensing view-point}. Clutters are not explicitly drawn to avoid congestion. {We neglect the multi-hop reflected channels for the sensing purpose because of the three/multi-fold cascaded path loss effect.}}\label{fig: system_model_sensing}
	\end{subfigure}
	\caption{System model: communication and sensing with a repeater, and its implications to both. Interference links are indicated by red arrows.}
	\vspace{-1.5\baselineskip}
\end{figure}

\section{System Model}
We consider a bi-static \gls{isac} system where the transmit \gls{bs}, with $\Nt$ antennas,  beamforms a precoded sensing signal towards a target hot-spot area while simultaneously serving $K$ \gls{dl} single-antenna users. The system also contains a dual-antenna repeater. The receiver \gls{bs}, with $\Nt$ antennas, performs target detection based on the reflected echoes from the target, also via the repeater. {The system model from the communication point of view is illustrated in~\cref{fig: system_model_comm} and the sensing point of view in~\cref{fig: system_model_sensing}.} 
The system entails several desired and interference channels, whose definitions are:
\begin{itemize}[noitemsep,topsep=0pt, align=left,leftmargin=*]
\item $\vf_{n}\mbbC{\Nt}$: \gls{dl} channel from the transmit \gls{bs} to the $n$th user.
\item $h_{n}\in\mathbb{C}$: channel from the transmit antenna of the repeater to the receive antenna of the $n$th user.
\item $\att\mbbC{\Nt}$: the channel from the transmitter to the target.
\item $\btt\mbbC{\Nt}$: the channel from the transmitter to the repeater.
\item $\atr\mbbC{\Nr}$: the channel from the target to the receive \gls{bs}.
\item $\btr\mbbC{\Nr}$: the channel from the repeater to the receive \gls{bs}.
\item $g\in\mathbb{C}$: the channel from the target to the repeater.
\item $\vG_{\mttB}\mbbC{\Nr\times \Nt}$: the channel from the transmit to receive \gls{bs}. Also, let $\hat{\vG}_{\mttB}$ and $\tilde{\vG}_{\mttB}$ be the estimate and estimation error of this direct link channel between the \glspl{bs}.
\item $\vC\mbbC{\Nr\times \Nt}$: the composite clutter channel. Specifically, the clutter channel consisting of the \gls{los} and \gls{nlos} paths caused by permanent obstacles can be measured beforehand and canceled~\cite{Behdad_Ozlem_TWC}. The unknown part corresponds to the \gls{nlos} paths caused by temporary obstacles, encased in $\vC$.
\end{itemize}
Apart from the clutter channel, we assume perfect \gls{csi} {at the receiver for target detection as a necessary first step to understand the overall system performance with an intuitive and tractable formulation}.\footnote{{The composite channel via the target is $\alpha\atr\att^{T}$. We assume only $\atr\att^{T}$ to be known, as the \gls{bs} beamforms to a specific sensing zone whose associated \gls{aoa} and \gls{aod} are pre-calibrated. However, $\alpha$ is unknown and will be estimated as a part of the \gls{glrt} later.}} %
Further, let $\nu\in\mathbb{C}$ be the repeater amplification factor. The target is characterized by the associated \gls{rcs}, denoted by $\alpha$, which follows $\alpha\cn{0,\sigma_{\mttT}^2}$, $\sigma_{\mttT}^2$ being the \gls{rcs} variance. The \gls{rcs} is unknown at the receiver and needs to be estimated. Finally, we adopt a model wherein target detection is performed in a given hotspot, and a larger area can be covered by sending sensing signals to such hotspot areas serially~\cite{Behdad_Ozlem_TWC}.

\section{Signalling Scheme}
Let $\vx\bsr{\tau}$ be the precoded \gls{dl} transmit signal, transmitted at $\tau$th slot, which can be written as:\footnote{All subsequent analysis is for an arbitrary subcarrier of the underlying \gls{ofdm} waveform~(or any monochromatic signal with a prefix to remove time-dispersion), without explicitly showing the subcarrier index, i.e., before IFFT and after FFT on the transmitter and the receiver sides; with cyclic-prefix length being larger than the maximum delay spread.}
\begin{align}
\vx[\tau]=&\sqrt{\rho}\bbr{\sum\limits_{n=1}^{K}\sqrt{\pi_{n}}\vp_{n}s_{n}[\tau]+\sqrt{\pi_{\mttT}}\vp_{\mttT}s_{\mttT}[\tau]},~\tau\in\bsr{\tau_{\mttL}},\label{eq: x}
\end{align}
where, $\rho$ is the total transmit power at the \gls{bs}, $\pi_{n}$ and $\pi_{\mttT}$ denote the fraction of the total power allocated to the $n$th \gls{dl} user and for sensing, respectively. Also, $\vp_{n}$ and $\vp_{\mttT}$ are the corresponding unit-normalized precoding vectors for the $n$th \gls{dl} user and the target. Similarly, $s_{n}\in\mathbb{C}$ and $s_{\mttT}\in\mathbb{C}$ are  unit-energy, statistically independent symbols. Also, $\Elr{\snorm{\vx[\tau]}}\leq \rho$, satisfying the total transmit power budget. Finally, $\tau_{\mttL}$ is the total duration over which sensing signals are transmitted, which can also be a coherence block. Given this, we can write the input-output receive~($y_{\mttR}^{\mttI}$) and transmit~($y_{\mttR}^{\mttO}$) signals at the repeater terminals as:
\begin{subequations}
\begin{align}
&y_{\mttR}^{\mttI}[\tau]=\alpha g \bbr{\att^{T}\vx[\tau]}+\btt^{T}\vx[\tau],\\
&y_{\mttR}^{\mttO}[\tau]=\nu\bpr{y_{\mttR}^{\mttI}[\tau]+w_{\mttR}[\tau]},
\end{align}
where $w_{\mttR}[\tau]\cn{0, \NvarR}$ is the \gls{awgn} at the repeater.
\end{subequations}
\subsubsection{Sensing with repeater}
The signal received at the \gls{bs} is
\begin{align*}
\vy_{\mathtt{BS}}[\tau]=\alpha\atr\att^{T}\vx[\tau]+\btr y_{\mttR}^{\mttO}[\tau]+\vG_{\mttB}\vx[\tau]+\vC\vx[\tau]+\vw_{\mttB}[\tau],
\end{align*}
where $\vw_{\mttB}[\tau]\cn{\vZ_{\Nr}, \NvarB\vI_{\Nr}}$ is the \gls{awgn} at the \gls{bs}. Knowledge of the transmitted signal $\vx[\tau]$ can be made available at the receiver \gls{bs} via back-haul links and thus be treated as pilot symbols. Then, the receive \gls{bs} can subtract $\hat{\vG}_{\mttB}\vx[\tau]$ from $\vy_{\mathtt{BS}}$ to reduce inter-\gls{bs} interference. After this, substituting for $y_{\mttR}^{\mttO}$, the desired signal for the target detection and associated interference terms become:
\begin{align}
\vy_{\mathtt{BS}}[\tau]=\vr[\tau]\alpha+\dot{\vC}\vx[\tau]+\dot{\vw}_{\mttB}[\tau],\label{eq: target_model}
\end{align}
where the equivalent sensing channel $\vR[\tau]$, equivalent clutter channel $\dot{\vC}$, and the overall sensing noise $\dot{\vw}_{\mttB}[\tau]$, are:
\begin{subequations}
\begin{align}
&\vr[\tau]=\bpr{\atr\att^{T}+\nu g \bbr{\btr\att^{T}}}\vx[\tau],\\
&\dot{\vC}={\vC+\nu\btr\btt^{T}},\\
&\dot{\vw}_{\mttB}[\tau]=\tilde{\vG}_{\mttB}\vx[\tau]+\bbr{\nu w_{\mttR}[\tau]}\btr+\vw_{\mttB}[\tau].
\end{align}
\end{subequations}
The direct (amplified) interference due to the repeater, i.e., $\bbr{\nu\btr\btt^{T}\vx[\tau]}$, is incorporated as a part of the clutter.  {However,  the composite channel can be pre-estimated (in \gls{los}, it can even be fully characterized by \glspl{aoa} and \glspl{aod}) to cancel this interference, similar to inter-\gls{bs} interference.\footnote{This can also be cancelled by {choosing} the precoder for target $\bpr{\vI_{\Nt}- \frac{1}{\snorm{\btt}}\btt^{*}\btt^{T}}\att,$ leading to the direct link interference $\nu\btr\btt^{T}\vp_{\mttT}=0$, where $\vp_{\mttT}$ is the precoder intended for the target. Although it is not necessary for the subsequent analyses.} The resulting $\dot{\vC}\approx\vC$  encases the effects of environmental scatterers.} Later, with~\eqref{eq: target_model}, we will develop the target detection framework. 
\subsubsection{\gls{dl} communication with repeater}
For \gls{dl} communication analysis, we drop the explicit dependence of index $\tau$. The signal received at the $n$th \gls{dl} user can be expressed as
\begin{align}
y_{\mathtt{UE}}=\vf_{n}^{T}\vx+h_{n}\bbr{\nu\bpr{\alpha g \bbr{\att^{T}\vx}+\btt^{T}\vx+w_{\mttR}}}+w_{\mathtt{UE}},\label{eq: received_UE1}
\end{align}
where the second term corresponds to the repeater transmission and $w_{\mathtt{UE}}\cn{0, \NvarUE}$ is  \gls{awgn} at the UE. Rearranging~\eqref{eq: received_UE1}, we get
\begin{align}
y_{\mathtt{UE}}=&\bpr{\vf_{n}^{T}+\nu h_{n}\btt^{T}}\vx+\dot{w}_{\mathtt{UE}},\label{eq: received_UE2}
\end{align}
with $\dot{w}_{\mathtt{UE}}=\bpr{\nu h_{n}w_{\mttR}+w_{\mathtt{UE}}}$.
\begin{rem}
{Note that the term $``\nu \alpha h_{n} g \bbr{\att^{T}\vx}"$ in~\eqref{eq: received_UE1} can be subsumed in the effective \gls{dl} channel of the users, i.e., $\bbr{\vf_{n}}$; and is therefore not explicitly shown in~\eqref{eq: received_UE2}. This is because the user is agnostic of the presence of the target, and any reflection from it can be treated as from any other environmental scatterers.}
\end{rem}


\section{\Gls{rcs} Estimation \& Target Detection}
Let $\overline{\vy}_{\mathtt{BS}}\triangleq\bsr{{\vy}_{\mathtt{BS}}^{T}[1], {\vy}_{\mathtt{BS}}^{T}[2], \ldots, {\vy}_{\mathtt{BS}}^{T}[\tau_{\mttL}]}^{T}\mbbC{\Nr\times \tau_{\mttL}}$ be the concatenated received signal vector at the \gls{bs} at the end of the observation window. With this, we can formulate the following binary hypothesis testing framework:
\begin{subequations}
\begin{align}
&\CalH_{0}:\quad \overline{\vy}_{\mathtt{BS}}= \overline{\vc}+\overline{\vw},
\\
&\CalH_{1}:\quad \overline{\vy}_{\mathtt{BS}}=\overline{\vr}+\overline{\vc}+\overline{\vw},
\end{align}
\end{subequations}
where $\overline{\vr}$, $\overline{\vc}$, and $\overline{\vw}$ are the concatenation of $\vr[\tau]\alpha$, $\dot{\vC}\vx[\tau]$, and $\dot{\vw}_{\mttB}[\tau]$, respectively. The null hypothesis $\CalH_{0}$ represents the case that there is no target in the sensing area, while the alternative hypothesis
$\CalH_{1}$ represents the existence of the target. The joint \gls{rcs} estimation, target-free channel~(\ie  clutter) estimation, and the hypothesis testing problem can be written as
$\bbr{\hat{\alpha}, \hat{\dot{\vc}}, \hat{\CalH}} = \argmax{\bbr{{\alpha}, {\dot{\vc}}, {\CalH}}}\quad \p{\alpha, \dot{\vc}, {\CalH} \;\vert\;\overline{\vy}_{\mathtt{BS}}},$%
where $\dot{\vc}\triangleq \vt{\dot{\vC}}$\footnote{``$\vt{.}$" stands for vectorization.} and $\p{\alpha, \dot{\vc}, {\CalH} \;\vert\;\overline{\vy}_{\mathtt{BS}}}$ is the joint \gls{pdf} given the received vector $\overline{\vy}_{\mathtt{BS}}$. Thus, the \gls{glrt} can be formulated as $\CalL \gldec{\CalH_{1}}{\CalH_{0}}\lambda_{\sfs}$, where 
\begin{align}
\CalL=\dfrac{\maximize{\alpha, \dot{\vc}}\quad \p{\overline{\vy}_{\mathtt{BS}} \;\vert\; \alpha, \dot{\vc}, \CalH_{1}} \p{\alpha \;\vert\; \CalH_{1}} \p{ \dot{\vc} \;\vert\; \CalH_{1}}}{\maximize{\dot{\vc}}\quad \p{\overline{\vy}_{\mathtt{BS}} \;\vert\; \dot{\vc}, \CalH_{0}} \p{ \dot{\vc} \;\vert\; \CalH_{0}}},\label{eq: L}
\end{align}
and $\lambda_{\sfs}$ is the threshold used by the detector, which is selected to achieve a desired false alarm probability. 
\begin{lem}\label{lem: GLRT_Repeater}
Let $\bbSigma_{\sfc}$ and $\bbSigma_{\sfs}[\tau]$ denote the covariance of the clutter channel and the overall sensing noise $\dot{\vw}_{\mttB}[\tau]$~(which are individually zero mean random vectors); and let $\CalT\triangleq \ln\CalL$ be the test statistics, which evaluates to
\begin{align}
\CalT=\mathbf{t}_{\CalH_{1}}^{H}\vQ_{\CalH_{1}}^{-1}\mathbf{t}_{\CalH_{1}}-\mathbf{t}_{\CalH_{0}}^{H}\vQ_{\CalH_{0}}^{-1}\mathbf{t}_{\CalH_{0}},
\end{align}
where $\mathbf{t}_{\CalH_{1}}$, $\mathbf{t}_{\CalH_{0}}$, $\vQ_{\CalH_{0}}$ are given by
\begin{subequations}
\begin{align}
&\mathbf{t}_{\CalH_{1}}=\begin{bmatrix}
\sum\limits_{\tau=1}^{\tau_{\mttL}}\vr^{H}[\tau]\bbSigma_{\sfs}^{-1}[\tau]\vy_{\mtt{BS}}[\tau]\\
\sum\limits_{\tau=1}^{\tau_{\mttL}}\vB^{H}[\tau]\bbSigma_{\sfs}^{-1}[\tau]\vy_{\mtt{BS}}[\tau]
\end{bmatrix},\label{eq: t1}\\
&\mathbf{t}_{\CalH_{0}}=\sum\limits_{\tau=1}^{\tau_{\mttL}}\vB^{H}[\tau]\bbSigma_{\sfs}^{-1}[\tau]\vy_{\mtt{BS}}[\tau], \label{eq: t0}\\
&\vQ_{\CalH_{0}}={\sum\limits_{\tau=1}^{\tau_{\mttL}}\vB^{H}[\tau]\bbSigma_{\sfs}^{-1}[\tau]\vB[\tau]+\bbSigma_{\sfc}^{-1}}.\label{eq: Q0}
\end{align}
\end{subequations}
$\vQ_{\CalH_{1}}$ is given in~\eqref{eq: Q} on the next page, and $\vB[\tau]\triangleq\bpr{\vx^{T}[\tau]\otimes \vI_{\Nr}}$.
The resulting test can be evaluated as
\begin{align}
\hat{\CalH}=\begin{cases}
\CalH_{1},\quad\text{if}\quad \CalT\geq \ln\lambda_{\sfs}^{'},\\
\CalH_{0},\quad\text{if}\quad \CalT<\ln\lambda_{\sfs}^{'}.
\end{cases}
\end{align}
with $\lambda_{\sfs}^{'}(=\pi\lambda_{\sfs}\sigma^2_{\mttT})$ being the detection threshold.
\begin{myfigure*}
\begin{align}
\vQ_{\CalH_{1}}=
\bsr{
\begin{array}{c@{\hspace{5pt}}:@{\hspace{5pt}}c}
{\sum\nolimits_{\tau=1}^{\tau_{\mttL}}\vr^{H}[\tau]\bbSigma_{\sfs}^{-1}[\tau]\vr[\tau]+\frac{1}{\sigma_{\mttT}^2}} & {\sum\nolimits_{\tau=1}^{\tau_{\mttL}}\vr^{H}[\tau]\bbSigma_{\sfs}^{-1}[\tau]\vB[\tau]} {\hspace{8pt}} \\ [3pt] \hdashline
\bpr{\sum\nolimits_{\tau=1}^{\tau_{\mttL}}\vr^{H}[\tau]\bbSigma_{\sfs}^{-1}[\tau]\vB[\tau]}^{H} &  {\hspace{2pt}}{\sum\nolimits_{\tau=1}^{\tau_{\mttL}}\vB^{H}[\tau]\bbSigma_{\sfs}^{-1}[\tau]\vB[\tau]+\bbSigma_{\sfc}^{-1}}\\ [2pt] 
\end{array}
}
\mbbC{(\Nt\Nr+1)\times (\Nt\Nr+1)}.\label{eq: Q}
\end{align} 
\end{myfigure*}
\end{lem}
\begin{proof}
See Appendix~\ref{app: A}. 
\end{proof}

{The complexity of the GLRT is dominated by a matrix inversion of dimension $(\Nt\Nr+1)\times (\Nt\Nr+1)$, and hence is $\mathcal{O}\left((\Nt\Nr+1)^3\right)$.}
Next, only the parameter remains to be computed in the sensing noise covariance $\bbSigma_{\sfs}[\tau]=\Elr{\dot{\vw}_{\mttB}[\tau]\dot{\vw}_{\mttB}^{H}[\tau]}$, which is evaluated~(treating transmitted symbols as pilots) as $\bbSigma_{\sfs}[\tau]=\Elr{\tilde{\vG}_{\mttB}\vx[\tau]\vx^{H}[\tau]\tilde{\vG}_{\mttB}^{H}}+\bbr{\abs{\nu}^2 \NvarR}\btr\btr^{H}+\NvarB\vI_{\Nr}.$
Next, if we assume that elements of $\tilde{\vG}_{\mttB}$ are \gls{iid} $\cno{0,\zeta^2}$, $\zeta^2$ capturing the power of the residual inter-BS interference~\cite{AC_CRM_TCOM_2024}, then $\Elr{\tilde{\vG}_{\mttB}\vx[\tau]\vx^{H}[\tau]\tilde{\vG}_{\mttB}^{H}}$ evaluates to $\zeta^2\snorm{\vx[\tau]}\vI_{\Nr}$, which completes the calculation of $\bbSigma_{\sfs}[\tau]$.
\begin{rem}
{We note that a noisy repeater can severely degrade the sensing performance, which is captured by $\abs{\nu}^2 \NvarR$ in the sensing noise covariance. The effect of the amplified \gls{dl} signals~(directly via the repeater) is implicitly captured in the clutter covariance matrix in~\eqref{eq: Q0}.}
\end{rem}
\section{\gls{dl} \gls{se} Analysis }
In \gls{dl}, based on~\eqref{eq: x} and~\eqref{eq: received_UE2}, the signal received at the $n$th user, with $\dot{\vf}_{n}\triangleq {\vf_{n}+\nu h_{n} \btt}$, can be expressed as:
\begin{align}
y_{\mathtt{UE}, n}=&\sqrt{\rho}\sqrt{\pi_{n}}\dot{\vf}_{n}^{T}\vp_{n}s_{n}+\sqrt{\rho}\sum\nolimits_{n'=1, n'\neq n}^{K}\sqrt{\pi_{n'}}\dot{\vf}_{n}^{T}\vp_{n'}s_{n'}\notag\\&+\sqrt{\rho}\sqrt{\pi_{\mttT}}\dot{\vf}_{n}^{T}\vp_{\mttT}s_{\mttT}+\dot{w}_{\mathtt{UE}},\label{eq: received_UEn}
\end{align}
where the second and the third term correspond to multi-user interference and interference due to a dedicated signal for target detection. The instantaneous \gls{sinr} at the $n$th \gls{dl} user can then be written as $\sinr_{n}=\frac{\rho\pi_{n}\abs{\dot{\vf}_{n}^T\vp_{n}}^2}{\rho\sum\limits_{n'=1, n'\neq n}^{K}\pi_{n'}\abs{\dot{\vf}_{n}^{T}\vp_{n'}}^2+\rho\pi_{\mttT}\abs{\dot{\vf}_{n}^{T}\vp_{\mttT}}^2+\Elr{\abs{\dot{w}_{\mathtt{UE}}}^2}},$
where the  noise variance can be computed as $\Elr{\abs{\dot{w}_{\mathtt{UE}}}^2}=\bbr{\abs{\nu}^2 \abs{h_{n}}^2\NvarR+\NvarUE}.$
Next, for the choice of the precoders, the transmitter can either consider the composite channel $\bbr{\dot{\vf}_{n}}$ or be completely agnostic of the repeater and consider $\bbr{{\vf}_{n}}$. We consider regularized \gls{zf} precoding scheme for the communication user, in which case $\vp_{n}$ becomes $\vp_{n}=\epsilon_{n}\bpr{\sum\nolimits_{n'=1}^{K}\dot{\vf}_{n'}\dot{\vf}_{n'}^{H}+\zeta_{\mtt{ZF}}\vI_{\Nt}}^{-1}\dot{\vf}_{n},$ where $\epsilon_{n}$ is the normalizer to ensure $\snorm{\vp_{n}}=1$ and $\zeta_{\mtt{ZF}}>0$ is a regularizer. For $\vp_{\mttT}$ we can consider one of the following:
    \emph{Target centric}: In this case, the transmitter beamforms towards the target, i.e., $\vp_{\mttT}=\epsilon_{\mttT}\dot{\vf}_{\mttT}$, with $\dot{\vf}_{\mttT}=\att$ and $\epsilon_{\mttT}$ is an appropriate normalizer.
    \emph{Communication centric}: We null the (potentially destructive) interference from the sensing signal to the users by projecting $\vp_{\mttT}$ onto the nullspace of the subspace spanned
by the users' channel vectors. In this case, $\vp_{\mttT}=\epsilon_{\mttT}\bpr{\vI_{\Nt}-\vU\vU^{H}}\dot{\vf}_{\mttT}$, where $\vU$ is a unitary matrix formed with the orthonormal columns
that span the column space of $\bbr{\dot{\vf}_{1}, \dot{\vf}_{2},\ldots, \dot{\vf}_{K}}$.

\section{Numerical Results \& Discussions}
In this section, we quantify the effects of a repeater on the overall \gls{isac} performance. The transmit power at the transmit \gls{bs} is $1$ Watt. The path loss for the communication channels is modeled using the $3$GPP Urban Microcell model. Each of the \glspl{bs} {is} equipped with $8$ antennas. 
The height of the \gls{bs}, repeaters, and users are $25$, $10$, and $1.5$ meters, respectively, following the models adopted in~\cite{Jianan_TWC}.
The carrier frequency and bandwidth are set to $1.9$ GHz and $20$ MHz, respectively. At the receive \gls{bs}, the noise spectral density is $-174$ dBm/Hz, and the noise figure is $9$ dB~\cite{Michail_TWC}. {The zone area is $500$ square meters, with the \glspl{bs} being in $(0, 250)$ and $(450, 250)$; the sensing zone is $(225, 250)$. The repeater is located at $(125, 250)$. The user locations are randomly generated.}
The repeater has the same noise power as the \gls{bs}. The sensing-related parameters follow the Swerling-I model~\cite{Wei_ISAC_Channel_Model} and $\bbr{\att, \atr}$ are generated following~\cite{Behdad_Ozlem_TWC}. 
{For clutter, \gls{iid} Gaussian taps are assumed.} 
Other relevant parameters are also mentioned alongside the results. The inter-\gls{bs} channel and the residual interference are modeled according to~\cite{AC_CRM_TCOM_2024}.


\begin{figure*}
\centering
\begin{minipage}{0.32\linewidth}
    \centering
    \includegraphics[width=.93\linewidth]{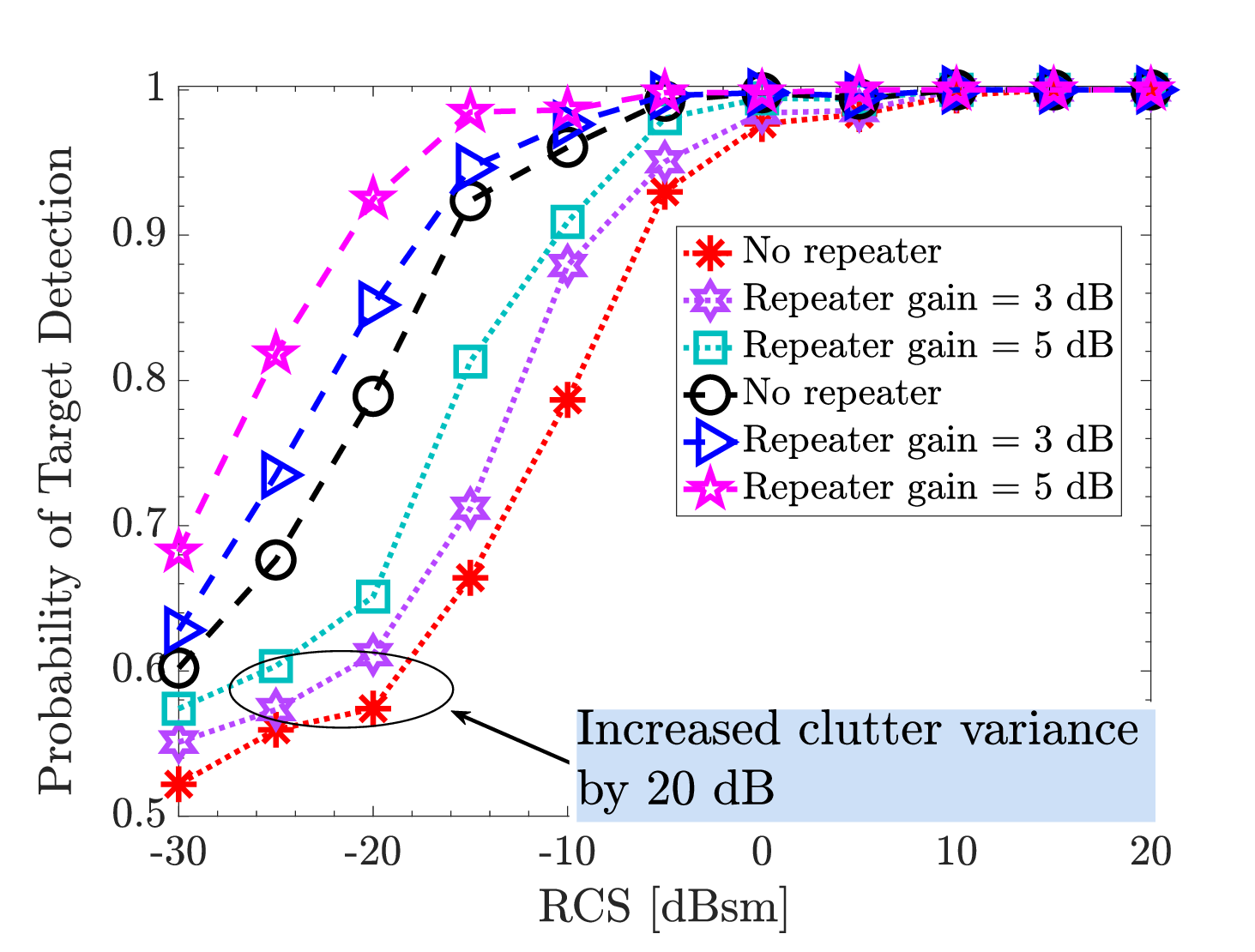}
    \caption{\Gls{pod} versus the target \gls{rcs}. Here, the repeater gains are measured with respect to the noise floor at the repeater. The threshold for \gls{glrt} is adjusted to maintain a \gls{pfa} of $0.01$.}\label{fig: 1}
 \end{minipage}
\hfill
\begin{minipage}{0.32\linewidth}
    \includegraphics[width =0.93\linewidth]{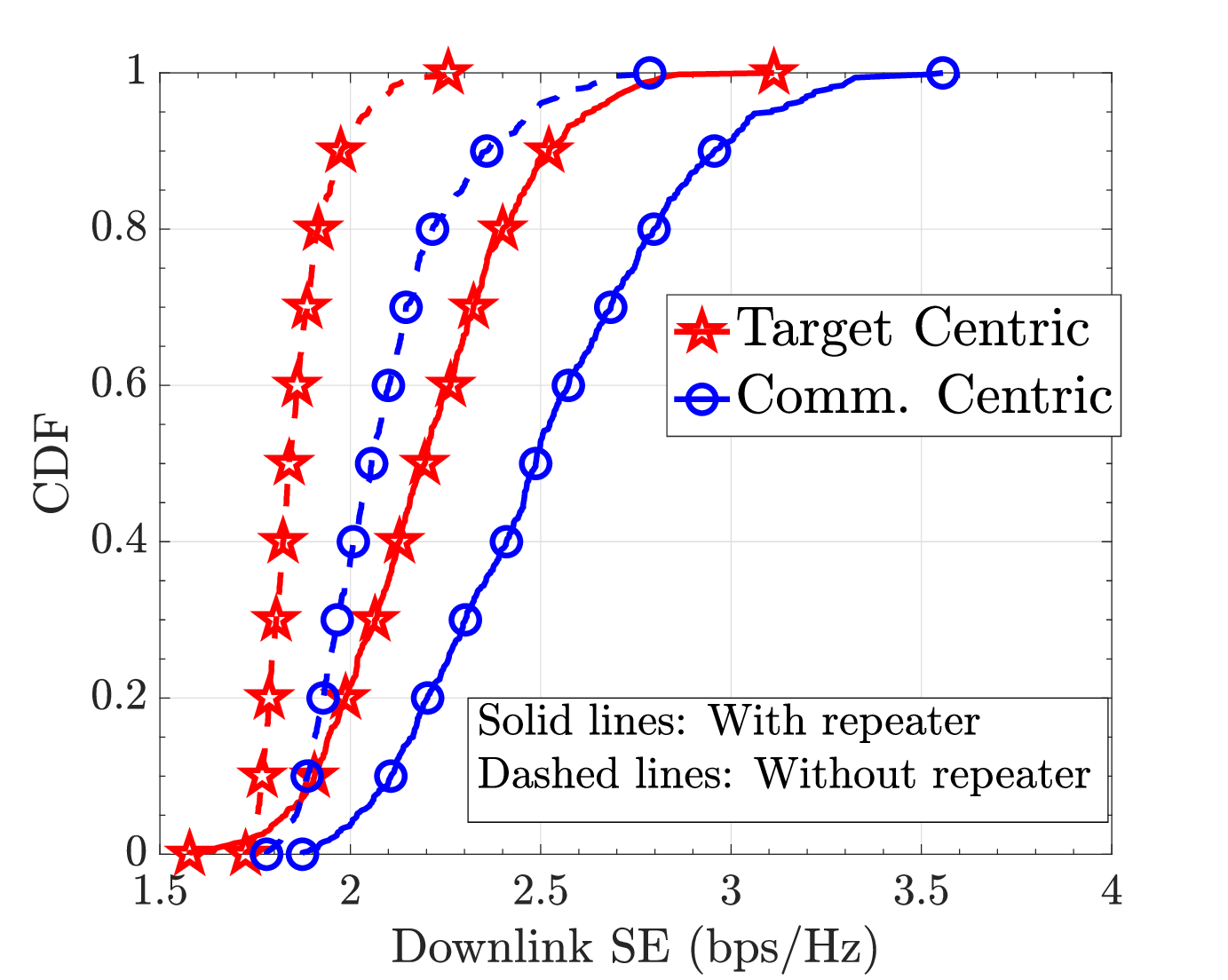}
\caption{\Gls{cdf} of \gls{dl} per-user \gls{se} with different choices of precoders. Equal power allocation has been employed by the transmit \gls{bs}.}\label{fig: 2}
\end{minipage}
\hfill
\begin{minipage}{0.32\linewidth}
    \centering
    \includegraphics[width=.9\linewidth]{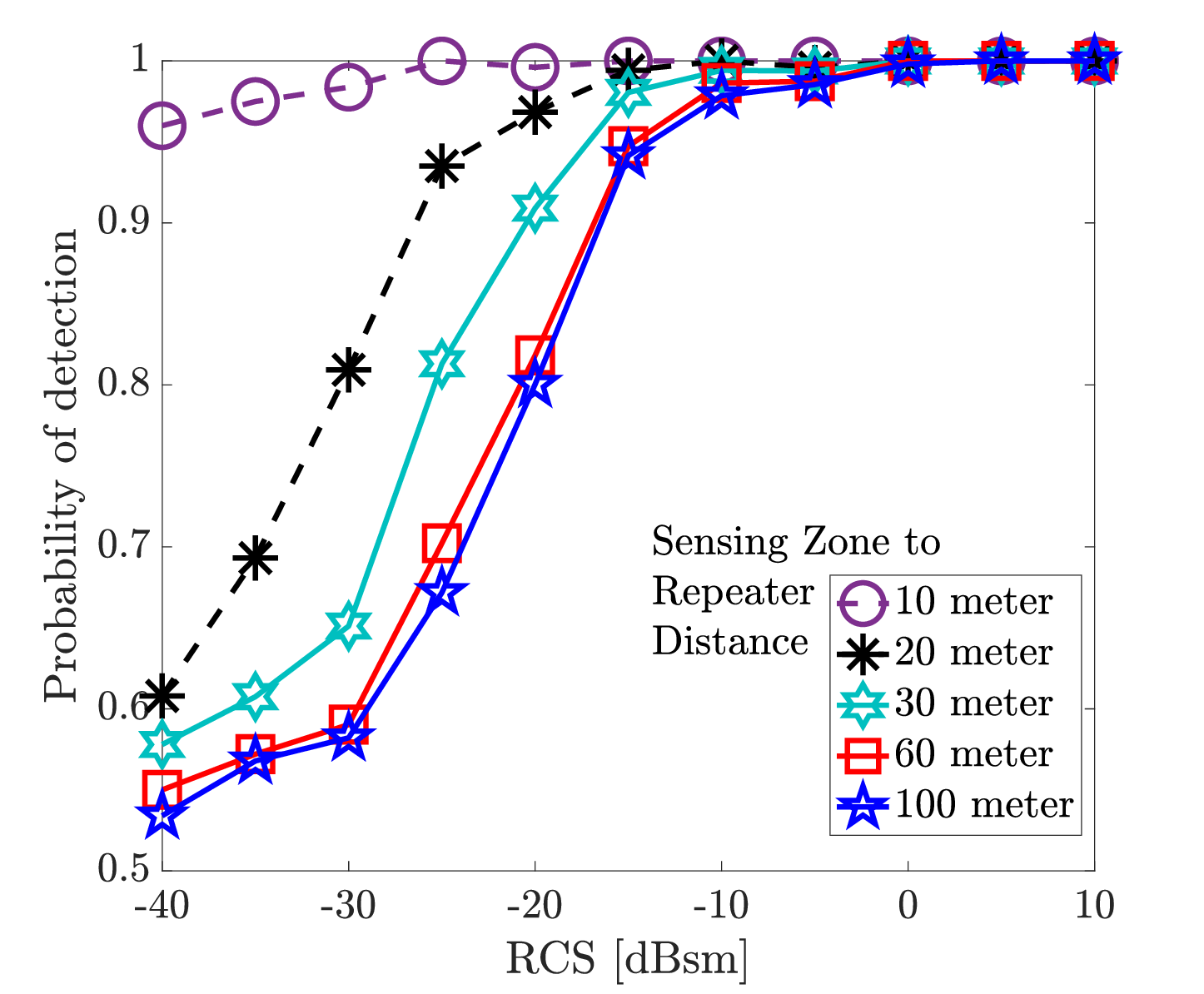}
    \caption{{Probability of detection with different repeater locations with respect to the target hotspot area. We assume the clutter variance to be $10$ dB below the noise floor and the \gls{pfa} is $0.01$.}}\label{fig: 4}
\end{minipage}
\end{figure*}

\Cref{fig: 1} illustrates the effects of a repeater on the target detection. For this, the repeater is placed within $100$ meters of the target hot-spot area, and we consider $K = 10$ communication users being served by the same \gls{bs}, whose locations are generated uniformly at random.  The slot duration for the \gls{glrt} consists of $50$ channel uses, \ie $\tau =1, 2, \ldots, 50$. 
For simplicity, we assume a Gaussian clutter model. Firstly, we observe that the repeater can substantially improve target detection performance as it amplifies and forwards the echo from the target and thereby improves the received echo strength at the \gls{bs}. Specifically, the repeater-aided system attains a high \gls{pod} at a lower value of \gls{rcs} compared to the traditional bi-static detection scenario. An increase in clutter variance degrades the overall detection performance; however, the repeater-assisted system remains robust to it in contrast with the no repeater case. 

In~\cref{fig: 2} illustrates the involved trade-offs between the different choices of precoders at the transmit \gls{bs}. {The repeater gain is set to  $5$ dB above the noise floor.} We can readily observe the degradation of the \gls{dl} \gls{se} for the target-centric precoder, as such a choice adds additional interference to the \gls{dl} users due to the dedicated target signals, for both with and without a repeater being present in the system. Communication-centric precoder nullifies that interference.
We also observe the benefits in the \gls{dl} \gls{se} due to the repeater. 

In~\Cref{fig: 4}, we illustrate the effects of moving the repeater towards the transmitter (varying the $x$-axis location) from the target hotspot area. This experiment reveals that a repeater near the hot-spot area can procure a high probability of detection value even for a lower target \gls{rcs}. The reflected echo from a repeater, which is also further from the target, suffers severely from pathloss and contributes little to the target detection. 

\edit{Finally,~\Cref{fig: trade-off} manifests the trade-off in the \gls{isac} performance, i.e., \gls{dl} \gls{se} versus \gls{pod}/\gls{pomd}, as the repeater location is varied from a zone where \gls{dl} users are located towards the target hot-spot area. Specifically, for this, we consider that $4$ downlink users are randomly located within $100$ meters of the transmit \gls{bs} and then the repeater location (x-axis) is moved from its original location~$(125, 250)$ towards the target hot-spot area. The RCS value is taken as $-30$ dBsm. The experiment demonstrates that a single repeater can aid communication or target detection based on its location, and thus motivates extending the current framework for a swarm-repeater-assisted system and investigating whether they can potentially offer uniformly good sensing and communication performance.}

\begin{figure}
\centering
\includegraphics[width = 0.7\linewidth]{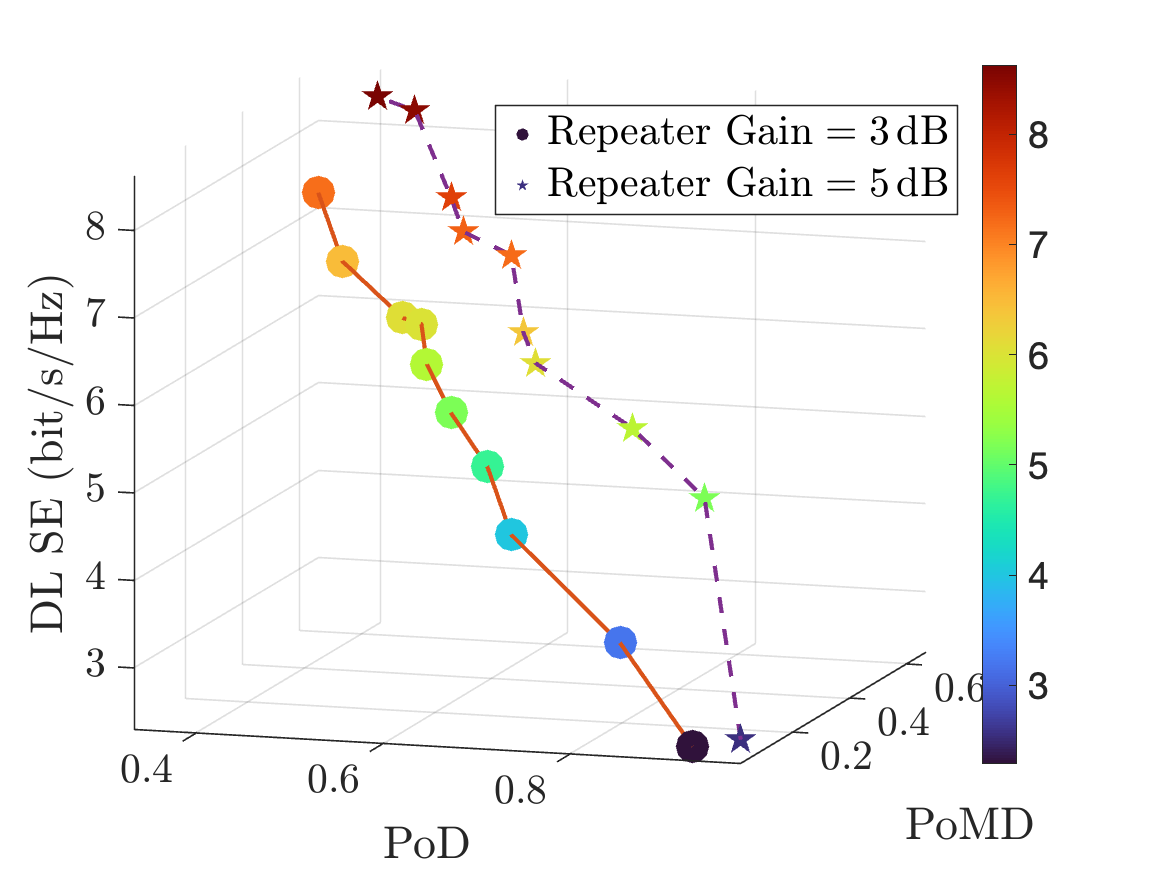}
\caption{\edit{\Gls{dl} \gls{se} versus target-detection (\Gls{pod} and \gls{pomd}) performance trade-off as the repeater location is moved from user locations to target hot-spot area.}}\label{fig: trade-off}
\end{figure}

\section{Remarks and Open Questions}
This paper presented initial results and a mathematical framework that facilitates repeater deployment for bi-static \gls{isac} and reported the benefits of using a repeater in the target hot-spot area.
One interesting future study is to incorporate channel estimation into the developed framework and to account for the time-varying nature of the channels. Furthermore, the impact of multiple repeaters~(\ie swarm-repeaters) in the system, including the effects of interaction between repeaters along with \gls{dl} power control, can be interesting to explore. 

\edit{Finally note that, \glspl{ris} can be a competing technology with repeaters. However, \gls{ris} and repeaters are governed by fundamentally different physics~(see~\cite[Table I]{Sara_Vieira_Erik_Mag} for a detailed comparison). Notably, the repeater is a band-selective device, whereas the \gls{ris} reflects all impinging signals regardless of their frequency. Therefore, repeater-assisted systems are not subject to out-of-band interference and are compatible with operational requirements in real-world multi-operator environments. Due to their ease of deployment, repeaters are already being actively considered for standardization in multi-operator systems. The vanilla repeater does not beamform, whereas an RIS can. It remains to be investigated how an $N$-element RIS  compares with a swarm of repeaters within a unified mathematical framework.} 
\appsection{Proof of Lemma~\ref{lem: GLRT_Repeater}}\label{app: A}
First, let us consider the optimization problem in the numerator of $\CalL$ in~\eqref{eq: L}. Recall, $\vB[\tau]\triangleq\bpr{\vx^{T}[\tau]\otimes \vI_{\Nr}}$, then $\dot{\vC}\vx[\tau]=\vB[\tau]\dot{\vc}$. Further, recall that $\p{\alpha \vert \CalH_{1}}$ and $\p{\dot{\vc} \vert \CalH_{1}}$ are respectively $\cno{0, \sigma_{\mttT}^2}$ and $\cno{\vZ_{\Nt\Nr}, \bbSigma_{\sfc}}$. Then, because of the temporal independence of the observation vectors, the numerator of $\CalL$ is equivalent to the following optimization 
\begin{align*}
&\maximize{\alpha, \dot{\vc}}\quad \bbr{\prod_{\tau=1}^{\tau_{\mttL}}\p{\vy_{\mathtt{BS}}[\tau] \;\vert\; \alpha, \dot{\vc}, \CalH_{1}}} \p{\alpha \;\vert\; \CalH_{1}} \p{ \dot{\vc} \;\vert\; \CalH_{1}}\\&=\maximize{\alpha, \dot{\vc}}\quad \bbr{\prod_{\tau=1}^{\tau_{\mttL}}\p{\dot{\vw}_{\mttB}[\tau]=\vy_{\mathtt{BS}}[\tau]-\vr[\tau]\alpha+\vB[\tau]\dot{\vc} \;\vert\; \alpha, \dot{\vc}}}\\&\hspace{12em}\times \p{\alpha \;\vert\; \CalH_{1}} \p{ \dot{\vc} \;\vert\; \CalH_{1}}\\&=\mttC_{\CalH_{1}} \exp \bbr{-\minimize{\alpha, \dot{\vc}} f(\alpha, \dot{\vc} ; \CalH_{1})},
\end{align*}
where $\mttC_{\CalH_{1}} $ is a constant and the function $f(\alpha, \dot{\vc} ; \CalH_{1})$, after some algebra, can be expressed as
\begin{align*}
\sum\nolimits_{\tau=1}^{\tau_{\mttL}}\vy_{\mathtt{BS}}^{H}[\tau]\bbSigma_{\sfs}^{-1}[\tau]\vy_{\mathtt{BS}}+\vz^{H}\vQ_{\CalH_{1}}\vz-2\Re\bbr{\vz^{H}\mathbf{t}_{\CalH_{1}}},
\end{align*}
where $\vz=\bsr{\alpha, \dot{\vc}^{T}}^T$ and the \gls{psd} matrix $\vQ_{\CalH_{1}}$ has the block-structure as shown in~\eqref{eq: Q}. Then, $\mathbf{t}_{\CalH_{1}}$ is expressed as shown in~\eqref{eq: t1}. Now, the estimate of $\vz$ obtained from $f(\alpha, \dot{\vc} ; \CalH_{1})$ equals to $\vQ_{\CalH_{1}}^{-1}\mathbf{t}_{\CalH_{1}}$, resulting 
\begin{align*}
-\minimize{\alpha, \dot{\vc}} f(\alpha, \dot{\vc} ; \CalH_{1})=-\sum\limits_{\tau=1}^{\tau_{\mttL}}\vy_{\mathtt{BS}}^{H}[\tau]\bbSigma_{\sfs}^{-1}[\tau]\vy_{\mathtt{BS}}+\mathbf{t}_{\CalH_{1}}^{H}\vQ_{\CalH_{1}}^{-1}\mathbf{t}_{\CalH_{1}}.
\end{align*}
Similarly, we can solve the optimization of the denominator of $\CalL$, and show that it evaluates to 
\begin{align*}
&\maximize{\dot{\vc}}\quad \p{\overline{\vy}_{\mathtt{BS}} \;\vert\; \dot{\vc}, \CalH_{0}} \p{ \dot{\vc} \;\vert\; \CalH_{0}}\notag\\&=\mttC_{\CalH_{0}} \exp\bbr{\mathbf{t}_{\CalH_{0}}^{H}\vQ_{\CalH_{0}}^{-1}\mathbf{t}_{\CalH_{0}}-\sum\nolimits_{\tau=1}^{\tau_{\mttL}}\vy_{\mathtt{BS}}^{H}[\tau]\bbSigma_{\sfs}^{-1}[\tau]\vy_{\mathtt{BS}}},
\end{align*}
where $\mttC_{\CalH_{0}} $ is a constant, $\mathbf{t}_{\CalH_{0}}$ and $\vQ_{\CalH_{0}}$ are respectively given by~\eqref{eq: t0} and~\eqref{eq: Q0}.
From the above, we get the final result. \qed

\ifCLASSOPTIONcaptionsoff
\newpage
\fi
\bibliographystyle{IEEEtran.bst}
\typeout{}
\bibliography{IEEEabrv.bib, Bibliography_list.bib}

\end{document}